\documentclass[runningheads]{llncs}
\usepackage{amsmath}
\usepackage{amssymb}
\usepackage{cite}
\usepackage{url}

\numberwithin{equation}{section}

\title{MEV in Multiple Concurrent Proposer Blockchains}
\author{
  Steven Landers\inst{1} \and
  Benjamin Marsh\inst{1, 2}
}

\institute{
  Sei Labs \and
  University of Portsmouth\\\email{\{steven, ben\}@seinetwork.io}
}

\authorrunning{S. Landers and B. Marsh}
\begin{document}
\maketitle

\begin{abstract}
    We analyze maximal extractable value in multiple concurrent proposer blockchains, where multiple blocks become data available before their final execution order is determined. This concurrency breaks the single builder assumption of sequential chains and introduces new MEV channels, including same tick duplicate steals, proposer to proposer auctions, and timing races driven by proof of availability latency. We develop a hazard normalized model of delay and inclusion, derive a closed form delay envelope \(M(\tau)\), and characterize equilibria for censorship, duplication, and auction games. We show how deterministic priority DAG scheduling and duplicate aware payouts neutralize same tick MEV while preserving throughput, identifying simple protocol configurations to mitigate MCP specific extraction without centralized builders.
\end{abstract}

\section{Introduction}
In this work we formalize several new MEV and collusion games introduced by multiple concurrent proposer (MCP) blockchains that are absent from sequential chains. We actively avoid discussing related problems such as TFMs and pricing in this work, and where auctions are discussed a generic EIP-1559 style mechanism with an inclusion price and an optional priority fee is assumed. As blockchains strive to scale, three clear paths have appeared in industry, ZK approaches, which today generally rely on a single sequencer, sharded approaches such as NEAR~\cite{skidanov2019nightshade}, and
more recently MCP approaches~\cite{garimidi2025multiple, marsh2025sei, NeuderResnick2024ConcurrentProposers, FoxResnick2023Multiplicity, Li2020Conflux}. In an MCP model each proposer in the network is able to propose 1 or more blocks in each tick of the system, as opposed to the single proposer in a traditional sequential model. This model allows for a higher block production rate, and a limited impact of an offline producer in any given tick of the system. We now describe, loosely, a simplified version of Sei Giga \cite{marsh2025sei} which will be used as an example of an MCP blockchain in this work. We assume an Autobahn \cite{giridharan2024autobahn} style BFT consensus protocol with a separate DD and consensus layer, see also \cite{Spiegelman2022Bullshark, Babel2024Mysticeti}. In such a protocol transactions are sent to RPC nodes where they are randomly allocated, uniformly, to any staked proposer in the network. There is no global mempool, though any proposer could run a local mempool. The proposer is then able to publish blocks as and when they desire to the DD layer. Once a block is on the DD layer it is gossiped to at least $f+1$ proposers to generate a proof of availability \cite{cohen2023proof} certificate. Each tick of the system, which happens every circa 300ms, a consensus proposal is made over the blocks with valid PoAs, without the blocks being gossiped to all voters. The blocks are executed after consensus with all blocks in the tick being executed as a single large block, to allow for dependencies between transactions sent to different proposers to be handled. Any invalid transactions, or duplicate transactions, are not executed and can be pruned by anyone but an archive node. In this work we assume an EIP-1559 style TFM with a posted price inclusion fee for execution which is burnt and a tip which is given the proposer. Because most DeFi flows involve tightly coupled multi-hop transactions, MCP concurrency introduces reaction surfaces that simply do not exist in single-builder mempools. We establish a hazard-normalized delay model for MCP MEV, deriving a closed-form delay envelope $M(\tau)$ and an immediate inclusion threshold $\tau^\dagger$, extending sequential chain delay models to the concurrent setting. We then formalize MCP specific MEV channels such as same tick duplicate steals, proposer-to-proposer auctions, PoA-latency timing races, and multi-submission externalities, and give explicit cutoff conditions for censorship, stealing, auctions, timing, and spam. Finally, we propose two protocol mechanisms that neutralize same-tick extraction, duplicate-aware tip splitting, which collapses the duplication prize, and a deterministic priority-DAG merge scheduler that preserves dependencies while enforcing tip priority. These mechanisms compress timing races and identify PoA rate as the core determinant of stealability, isolating which MEV channels are intrinsic to MCP and which can be pushed to $\varepsilon$ through straightforward protocol choices. Same tick ordering affects not only classic AMM arbitrage but also cascading DeFi events such as oracle writes, funding rate adjustments, and multi-venue liquidation waterfalls.
\subsection{MEV}
MEV refers to the additional profit a participant with influence over transaction ordering, inclusion, or censorship can obtain beyond posted tips and base fees. In transparent, first come first served mempools, this influence generates a market for both order flow and order manipulation, while also creating incentives for censorship and chain reorganization. Foundational work formalized the role of priority gas auctions and highlighted the consensus-layer risks introduced by transaction-order dependencies on decentralized exchanges~\cite{Daian2019FlashBoys}. Subsequent empirical analyses quantified the magnitude of MEV opportunities across arbitrage, liquidations, and sandwich trading~\cite{Qin2022DarkForest,Zhou2020HFTDEX}. On sequential blockchains, MEV arises primarily through ordering control (such as front-running, back-running, and sandwiching), censorship of rival or victim transactions, and opportunistic history revision through reorganization or time-bandit attacks~\cite{Carlsten2016Instability}. While MEV can occur both within and across blocks, the existing literature has largely focused on intra-block mechanisms. In all cases, MEV imposes externalities on both users and the underlying protocol. Private order flow channels and order flow auctions (OFAs) attempt to rebate value to users and reduce public mempool exposure, but also raise centralization concerns~\cite{FlashbotsOrderFlow2024,ConsensysOFA}. Orthogonal lines use fair‑ordering consensus to constrain how much advantage ordering control yields~\cite{Kelkar2020Aequitas,Kelkar2023Themis,Cachin2023QOF}. Another direction is encrypted mempools to hide payloads until commitment, reducing same tick informational edge~\cite{Choudhuri2024BatchedTE,ShutterDocs}.
\section{MEV in MCP}
MEV in sequential blockchains has been well studied in literature and a number of proposals have been made to limit MEV. In this work we isolate some MCP specific externalities arising from concurrent proposal and deterministic merge ordering that have no analogue in sequential single proposer chains. MEV in sequential blockchains can take many forms but is most frequently seen in the form of censorship, front-running, back-running, or sandwich attacks. These attacks are typically considered to take place within a single block, though some work has studied MEV across sequential blocks. Such attacks continue to exist in an MCP blockchain and the strategy space remains unchanged for both tx issuing agents and proposers attempting to extract MEV. Atop the more traditional scope of MEV, the use of an MCP protocol introduces several novel types of MEV which are either limited in scope or non-existent in sequential blockchains.

Though trivially present in a sequential blockchain it is worth briefly discussing censorship in the context of MCP. Once a transaction is received by a proposer it can arbitrarily choose whether to include any given transaction. If the proposer chooses not to include a transaction they lose the tip for that transaction but gain whatever expected utility censorship gave them. A Sei Giga style MCP protocol offers probabilistic censorship resistance however since any agent can issue any transaction multiple times for inclusion, therefore the proposer is acting in a state of incomplete information as the utility gained from censoring may be lost if another proposer has the transaction and includes it in the same tick of the system.

In a traditional sequential blockchain mempools exist and gossip transactions ensuring all proposers have access, at least ignoring private mempools, to the same set of transactions from which to build their block. Sei Giga has no mempool. This means a transaction, may, be known by only a single proposer who is able to auction off that transaction to another proposer who may have the ability to extract some MEV style side utility from said transaction.
Once a block is published on the DD layer its content is available for the public to see, if another proposer who sees the block sees a transaction they can extract utility from, either in expectation from their version of the tx being executed first giving them the tip, or via MEV, then they have an incentive to include that transaction in their own block.

Assume two users intend to bid on the same asset, but user 2 knows they have access to a faster block publisher. User 2 is able to wait for user 1 to submit their bid, for that block to be published on the DD layer before they submit their bid. In a sequential blockchain a similar game is possible at the mempool level but ensuring user 2's transaction is included first in the block requires a side payment to an arbitrary block builder, or a high tip and the hope that the tip acts as an ordering auction allowing that transaction to be included first. In MCP no such payments are required if it is known that the block will be included in the same tick of the system, after which ordering is deterministic.

The traditional MEV strategies from sequential models are extended in the MCP model to extend over multiple blocks in the same tick of the system. As such it is possible to attempt to perform an MEV game against a transaction in another block once it is public on the DD layer. 
Spam comes in many forms, but we focus on two channels in this work, namely, DoS-style use of the DD layer and attempts to gain access to the front of the merged block through spamming transactions.
\section{A model of MEV}
Notation used in this model is summarized in Appendix \ref{sec:not}. Let $F_x(\alpha)=A(1-e^{-k\alpha})$ be gross MEV accrual and $S_x(\alpha)=e^{-\lambda \alpha}$ the survival.
Then
\[
  \Delta_x(\alpha)=\int_{0}^{\alpha} A k e^{-k s}\,S_x(s)\,ds
  = \frac{A k}{k+\lambda}\Bigl(1-e^{-(k+\lambda)\alpha}\Bigr),
\]
and the per‑tx objective for delay $\alpha\ge 0$ is
\(
  U_{\rm mev}(\alpha,\tau_x)
  = \Delta_x(\alpha)\;+\; \tau_x\, e^{-\lambda \alpha}.
\)
We adopt a single hazard $\lambda>0$ that governs both the survival of an opportunity against pre-emption and the effective discounting of the tip. As such, the tip is realized only if $x$ survives until execution, hence the expected tip conditional on delay $\alpha$ is $\tau_x\,S_x(\alpha)=\tau_x e^{-\lambda \alpha}$. Additional time discounting over sub-second ticks ($\sim$300ms) is negligible relative to pre-emption risk, and inventory/credit frictions for proposers at this horizon are second order. The two hazard generalization with a distinct tip discount $\delta$ is given in Appendix~\ref{app:twohazard}, all main results here use the one hazard specialization $\delta=\lambda$.

Assume a set of proposers \(P = \{1, 2, \dots, n\}\). Let time be divided into discrete ticks \(T = 1, 2, \dots\). We define a transaction \(x\) as having a true bid \(f_x\) and a tip \(\tau_x\). For proposer \(i\) in tick \(t\), let \(X_i^t\) be the set of transactions it has received and let \(B_i^t \subseteq X_i^t\) be the (multi)set of transactions it actually includes in its PoA–certified blocks during tick \(t\).

We take a block-building baseline in which all transactions are included immediately:
\[
  u_{\text{build}}^{t,i}(B_i^t)
  \;=\;
  \sum_{x\in B_i^t} \tau_x.
\]

Upon receiving a transaction $x$, proposer $i$ chooses a delay $\alpha_x\ge 0$, where $\alpha_x=0$ is immediate inclusion and $\alpha_x>0$ harvests MEV. We allow a different delay for each transaction and collect these into a vector
\[
  \alpha^{t,i} \;=\; (\alpha_x)_{x\in B_i^t}.
\]

For a single transaction \(x\), the expected MEV from delaying by \(\alpha_x\) is
\[
  \Delta_x(\alpha_x)
  \;=\;
  \frac{A k}{k+\lambda}\Bigl(1-e^{-(k+\lambda)\alpha_x}\Bigr),
\]
and the expected tip conditional on survival is \(\tau_x e^{-\lambda \alpha_x}\). Thus the per‑transaction payoff from choosing delay \(\alpha_x\) is
\[
  u_x(\alpha_x)
  \;:=\;
  U_{\rm mev}(\alpha_x,\tau_x)
  \;=\;
  \Delta_x(\alpha_x) + \tau_x e^{-\lambda \alpha_x}.
\]

The side (MEV) component for proposer \(i\) in tick \(t\) is
\[
  u_{\text{mev}}^{t,i}(B_i^t,\alpha^{t,i})
  \;=\;
  \sum_{x\in B_i^t} \Delta_x(\alpha_x),
\]
while the realized tip component is
\[
  u_{\text{tip}}^{t,i}(B_i^t,\alpha^{t,i})
  \;=\;
  \sum_{x\in B_i^t} \tau_x e^{-\lambda \alpha_x}.
\]

We note that \(U_{\rm mev}(\alpha,\tau_x)\) is a per-transaction payoff combining both the MEV component \(\Delta_x(\alpha)\) and the hazard discounted tip \(\tau_x e^{-\lambda\alpha}\). At the block level we will sometimes refer to the pure-MEV component
\(
  u_{\rm mev}^{t,i}(B_i^t,\alpha^{t,i}) = \sum_{x\in B_i^t} \Delta_x(\alpha_x),
\)
which, unlike \(U_{\rm mev}\), does not include tips.

Putting these together, the total utility of proposer \(i\) in tick \(t\) is
\[
  U^{t,i}(B_i^t,\alpha^{t,i})
  \;=\;
  \sum_{x\in B_i^t} u_x(\alpha_x)
  \;=\;
  u_{\text{mev}}^{t,i}(B_i^t,\alpha^{t,i})
  \;+\;
  u_{\text{tip}}^{t,i}(B_i^t,\alpha^{t,i}).
\]
Immediate inclusion is $\alpha_x=0$, since $U_{\rm mev}(0,\tau_x)=\tau_x$.
To avoid clutter we will usually suppress the arguments and just write \(U^{t,i}\).
For a single transaction with tip \(\tau_x\), the proposer’s one-dimensional delay problem is
\[
  \max_{\alpha\ge 0} U_{\rm mev}(\alpha,\tau_x),
\]
and we define the per‑tx delay envelope
\[
  M(\tau)
  \;:=\;
  \sup_{\alpha\ge 0} U_{\rm mev}(\alpha,\tau),
\]
which depends on \(x\) only through its tip \(\tau_x\) and the MEV parameters \((A,k,\lambda)\).
\subsection{Censorship}
We will now model censorship in MCP from both the perspective of a transaction issuing agent, the user, and the proposer, following \cite{FoxPaiResnick2023CensorshipAuctions}. Upon receiving a transaction \(x\), the proposer faces four strategic options. It may include immediately, yielding utility \(u_{\rm inc}=\tau_x\), drop it entirely for a side payment \(\beta\), giving \(u_{\rm drop}=\beta\), delay inclusion to extract MEV over a time window \(\alpha\), obtaining
\(
  u_{\rm keep}(\alpha)
  \;:=\;
  U_{\rm mev}(\alpha,\tau_x)
  = \frac{A k}{k+\lambda}\bigl(1-e^{-(k+\lambda)\alpha}\bigr)
  + e^{-\lambda\alpha}\tau_x;
\)
or hold it for a proposer to proposer auction, where each of the \(m=|P|-1\) rival bidders \(j\) has valuation
\(
  V_j^x(\alpha)
  = \max\!\Bigl\{e^{-\lambda\alpha}\tau_x+\Delta_{x,j}(\alpha),\,\beta_j\Bigr\},
\)
and the seller’s expected revenue under the optimal auction with regular i.i.d.\ values \(F_\alpha\) is
\(
  u_{\rm auc}(\alpha)={\rm Rev}^*(\alpha),
\)
as established in Proposition~\ref{prop:myerson-iid} (in general \({\rm Rev}^*(\alpha)\neq \mathbb{E}[\max_j V_j^x(\alpha)]\)).

\begin{lemma}[Drop vs keep cutoff]\label{lem:drop-cutoff}
Let $U_{\mathrm{mev}}(\alpha,\tau)
= \frac{A k}{k+\lambda}\!\left(1-e^{-(k+\lambda)\alpha}\right)
+ \tau\, e^{-\lambda\alpha}$ and define
$M(\tau):=\sup_{\alpha\ge 0} U_{\mathrm{mev}}(\alpha,\tau)$.
Then $M(\tau)$ is continuous and strictly increasing with
$M(0)=\tfrac{A k}{k+\lambda}$ and $\lim_{\tau\to\infty} M(\tau)=\infty$.
Hence there exists a unique cutoff $\tau_d\ge 0$ solving $M(\tau_d)=\beta$.
The optimal proposer policy is
\[
a_i(x)=
\begin{cases}
\mathrm{DROP}, & \tau_x<\tau_d,\\[2mm]
\textsf{keep}(\alpha^\star(\tau_x)), & \tau_x\ge\tau_d,
\end{cases}
\]
where $\alpha^\star(\tau)$ equals $\tfrac{1}{k}\ln\!\tfrac{A k}{\lambda\tau}$ if $\tau< A k/\lambda$ and $0$ otherwise.
If $\beta\le M(0)=\tfrac{A k}{k+\lambda}$, then $\tau_d=0$ and no drop region exists.
\end{lemma}

\begin{lemma}[Closed form of $M(\tau)$]
Let $r:=\frac{\lambda\tau}{A k}$ with $k,\lambda>0$. Then
\[
M(\tau)=
\begin{cases}
\displaystyle
\frac{A k}{k+\lambda}\Bigl(1-r^{\,1+\lambda/k}\Bigr)+\tau\, r^{\,\lambda/k}, & 0<r<1,\\[0.8ex]
\displaystyle \max\!\left\{\tau,\ \frac{A k}{k+\lambda}\right\}, & \text{otherwise.}
\end{cases}
\]
In particular, $M(0)=\frac{A k}{k+\lambda}$ and $M(\tau)$ is strictly increasing in $\tau$.
\end{lemma}

Because $U_{\rm mev}(\alpha,\tau)\to \tfrac{A k}{k+\lambda}$ as $\alpha\to\infty$ and $U_{\rm mev}(0,\tau)=\tau$, and since $U_{\rm mev}$ is unimodal in $\alpha$ (Lemma~\ref{lem:max}), the maximizer is either the interior point $\alpha^\star(\tau)=\tfrac{1}{k}\ln\!\tfrac{A k}{\lambda\tau}$ when $\tau< A k/\lambda$, or the boundary $\alpha=0$ when $\tau\ge A k/\lambda$. Thus the immediate-inclusion threshold is
\[
  \tau^\dagger \;=\; \frac{A k}{\lambda},
\]
which separates delay ($\tau<\tau^\dagger$) from immediate inclusion ($\tau\ge\tau^\dagger$).

\begin{lemma}[Maximizer]\label{lem:max}
Let $\alpha^\star(\tau)$ denote any maximizer of $U_{\rm mev}(\alpha,\tau)$. With $U_{\rm mev}(\alpha,\tau)=\frac{A k}{k+\lambda}\!\left(1-e^{-(k+\lambda)\alpha}\right)+\tau e^{-\lambda\alpha}$ and $k,\lambda>0$,
\[
  \frac{d}{d\alpha}U_{\rm mev}(\alpha,\tau)
  \;=\; A k\, e^{-(k+\lambda)\alpha} \;-\; \lambda\,\tau\, e^{-\lambda \alpha}.
\]
Let $r:=\frac{\lambda\tau}{A k}$. Then:
\begin{itemize}
\item If $0<r<1$ (i.e.\ $\tau< A k/\lambda$), there is a unique interior maximizer
\[
  \alpha^\star(\tau) \;=\; \frac{1}{k}\,\ln\!\frac{1}{r}
  \;=\; \frac{1}{k}\,\ln\!\frac{A k}{\lambda\tau}\;>\;0.
\]
\item If $r\ge 1$ (i.e.\ $\tau\ge A k/\lambda$), the global maximizer is the boundary $\alpha^\star(\tau)=0$ (immediate inclusion).
\item If $\tau=0$ ($r=0$), the maximizer is attained in the limit $\alpha^\star\to+\infty$ with value $\frac{A k}{k+\lambda}$.
\end{itemize}
\end{lemma}

\begin{lemma}[Symmetric cutoff]\label{lem:SNE}
If every proposer uses the cutoff strategy
\[
a_i(x)=
\begin{cases}
\text{\textsc{DROP}},&\tau_x<\tau_d,\\
\text{choose }\alpha^\star(\tau_x)\text{ from Lemma~\ref{lem:max}},&\tau_x\ge\tau_d,
\end{cases}
\]
with $\tau_d$ from Lemma~\ref{lem:drop-cutoff}, then no single proposer can profitably deviate, this is a best response.
\end{lemma}

Users are impatient: utility is \(v>0\) if included in this tick and \(0\) otherwise.
Let the contacted proposer set be \(P_s\subseteq P\) with size \(s=|P_s|\).
For \(i\in P_s\), let censor probability be \(p_x^i\in[0,1)\). One broadcast round to all of \(P_s\) yields inclusion
probability \(1-\prod_{i\in P_s}p_x^i\). If the user repeats this broadcast to the same \(P_s\) for \(k\in\mathbb N_0\) rounds,
with per-submission cost \(c>0\), then with \(\Pi:=\prod_{i\in P_s}p_x^i\),
\[
  U^{(k)} \;=\; v\bigl(1-\Pi^k\bigr) - k\,s\,c.
\]
\begin{lemma}\label{lem:five}
Let \(Q\in[0,1)\) and \(s=|P_s|\). For \(k\in\mathbb N_0\), write
\[
  U^{(k)} \;=\; v\bigl(1-Q^k\bigr) - k\,s\,c.
\]
If \(v(1-Q)\le s\,c\), then \(\arg\max_{k\in\mathbb N_0} U^{(k)} = \{0\}\).
If \(v(1-Q)>s\,c\), define
\[
  k^* \;=\; \max\{k\in\mathbb N_0: v\,Q^k(1-Q) > s\,c\}.
\]
Then \(\arg\max_{k\in\mathbb N_0} U^{(k)} = \{k^*+1\}\).
\end{lemma}
Thus we can trivially see that despite the risk of censorship, as in sequential chains, the user in an MCP chain frequently has a best response that includes submitting the transaction multiple times for probabilistic censorship resistance, though in the real world the user is unlikely to know the real probability of censorship and must act on a simpler heuristic closer to \(k = \lfloor \frac{v}{c} \rfloor\).
\subsection{Tx stealing}
We formalize the decision of a rival proposer $j\neq i$ to steal a transaction $x$ that first appears publicly inside proposer $i$'s DA-published block during tick $t$.

Let $\sigma_i^t(x)\in[0,1]$ denote the stealability of $x$ against $i$ within the same tick, the probability that (i) $j$ observes $x$ early enough to get a valid PoA for a block containing a duplicate of $x$ in tick $t$, and (ii) the canonical inter-block ordering executes $j$'s duplicate of $x$ before $i$'s copy. Let $\rho_i^t(x)\in[0,1]$ denote the probability that $i$'s block misses tick $t$, while $j$ can still include $x$ in $t$ after seeing it, in that case the duplicate from $j$ executes and $j$ receives $\tau_x$ with probability $\rho_i^t(x)$.

Including a duplicate of $x$ costs proposer $j$ the shadow value of a marginal slot in $j$'s block during tick $t$, denoted $\phi_j^t\ge0$. If ordering before $i$ also enables an additional inter-block MEV gain on $x$, denote that incremental advantage by $\delta_x\ge0$.

\begin{proposition}[Steal threshold]\label{prop:steal-threshold}
Fix a transaction \(x\) first revealed by proposer \(i\) in tick \(t\).
Let \(\sigma_i^t(x)\in[0,1]\) be the probability that both \(i\) and a rival \(j\) make tick \(t\) but \(j\)'s duplicate executes before \(i\)'s;
let \(\rho_i^t(x)\in[0,1]\) be the probability that \(i\) misses tick \(t\) while \(j\) still includes \(x\) in \(t\).
These events are disjoint by construction. Let \(\phi_j^t\ge0\) be \(j\)'s shadow blockspace value and \(\delta_x\ge0\) a same-tick inter-block MEV increment.
Then \(j\) has a (weakly) profitable steal iff
\[
  \bigl(\sigma_i^t(x)+\rho_i^t(x)\bigr)\,\bigl(\tau_x+\delta_x\bigr) \;>\; \phi_j^t,
\]
equivalently \( \tau_x > \phi_j^t/(\sigma_i^t(x)+\rho_i^t(x)) - \delta_x\).
\end{proposition}

If multiple rivals may steal, the game becomes an all-pay contest for the unique prize $\tau_x+\delta_x$ \cite{baye1996allpay,tullock1980efficient}.

\begin{lemma}[Symmetric mixed equilibrium of steal attempts]\label{lem:steal-mixed}
Let there be $m\ge1$ potential thieves with identical parameters $(\sigma,\rho,\phi)$ and prize $R:=\tau_x+\delta_x$. Suppose each thief who attempts pays cost $\phi$ (consumes a slot) and, if $n\ge1$ thieves attempt, the winner is selected uniformly from the $n$ attempters (others' duplicates are pruned). Then there exists a unique symmetric mixed equilibrium \cite{baye1996allpay} in which each thief attempts with probability $p^*\in[0,1]$ characterized by
\(
  g(p^*) \;=\; 
  \sum_{n=0}^{m-1}\binom{m-1}{n} (p^*)^{n} (1-p^*)^{m-1-n}
  \left(\frac{(\sigma+\rho)\,R}{n+1}-\phi\right)
  \;=\; 0.
\)
Moreover, $p^*$ is increasing in $R$ and in $(\sigma+\rho)$, and decreasing in $\phi$ and $m$. 
\end{lemma}

\begin{corollary}[Anti-steal via equal tip-splitting]\label{cor:equal-split}
If duplicates of \(x\) that appear in tick \(t\) (the original plus \(n\) thieves) split the tip equally,
the per-copy prize is \(R_{n}=\tau_x/(n+1)\). For any \(\phi>0\) there exists \(\bar n\) such that
for all \(n\ge\bar n\), \(R_n<\phi\) and stealing is strictly unprofitable.
\end{corollary}

In sequential chains the analogous strategy requires being (or bribing) the sole block builder for the height. MCP exposes $x$ on DD before final ordering, creating a same-tick duplicate race that is weak or absent in sequential settings. If same-tick incremental MEV $\delta_x$ exists, the conclusion holds either when $\delta_x=0$ or when the mechanism splits $\delta_x$ proportionally across duplicates. Let $m_x$ denote the multiplicity of a logical transaction $x$ across $\mathcal B_t$ in tick $t$ (counting all PoA certified copies revealed before the tick boundary). Execute only one copy of $x$, but split the posted tip evenly across all $m_x$ in‑tick copies, each proposer of a copy receives $\tau_x/m_x$. This preserves correct state while implementing the equal‑split payoff used in Corollary~\ref{cor:equal-split}.

\subsection{Tx auctions}
A proposer $i$ who privately observes $x$ (no global mempool) may sell access to $x$ before DD reveal. There are $m\ge1$ bidders $j\in B=P\setminus\{i\}$. Fix delay $\alpha\ge0$. Each bidder’s payoff is
\[
V_j^x(\alpha) \;=\; e^{-\lambda\alpha}\,\tau_x \;+\; \Delta_{x}(\alpha),
\]
where $\Delta_{x}(\alpha)$ is the MEV accrual from with parameters $(A,k,\lambda)$. We assume i.i.d.\ regular values, $V_j^x(\alpha)\sim F_\alpha$ with density $f_\alpha$ on $[0,\bar v_\alpha]$ and increasing virtual value $\phi_\alpha(v)=v-\frac{1-F_\alpha(v)}{f_\alpha(v)}$.

\begin{proposition}[Myerson reserve under i.i.d.\ values]
\label{prop:myerson-iid}
Under the i.i.d.\ regular assumption, the revenue maximizing one-shot auction at delay~$\alpha$
uses reserve~$r^{\ast}(\alpha)$ satisfying $\phi_\alpha(r^{\ast}) = 0$.
The expected revenue equals
\[
{\rm Rev}^{\ast}(\alpha)
  = \mathbb{E}\!\left[\max\bigl\{\phi_\alpha(V^{\ast}),\,0\bigr\}\right],
  \qquad
  V^{\ast} = \max_{j\in B} V_j^x(\alpha).
\]
\end{proposition}

\noindent
If $V_j^x(\alpha) \sim \mathrm{Unif}[0,\bar v_\alpha]$, then $r^{\ast}(\alpha) = \bar v_\alpha/2$, and for a second–price auction with reserve,
\[
{\rm Rev}^{\ast}(\alpha)
  = \bar v_\alpha\!\left(\frac{m-1}{m+1} + \frac{2^{-m}}{m+1}\right).
\]

\begin{lemma}[Auction vs internalization cutoff]\label{lem:auction-cutoff}
Let $M(\tau)=\sup_{\alpha\ge0}U_{\rm mev}(\alpha,\tau)$.
The seller’s optimal action between \{\textsc{KEEP} (choose $\alpha\ge0$), \textsc{AUC}\} is
\[
  \arg\max\Bigl\{\,M(\tau_x),\ \ \sup_{\alpha\ge0} {\rm Rev}^*(\alpha)\,\Bigr\},
\]
where ${\rm Rev}^*(\alpha)$ is the Myerson optimal expected revenue for the $\alpha$‑parametrized value distribution. Immediate inclusion is $\alpha=0$ inside $M(\tau_x)$.
\end{lemma}

\begin{example}[Uniform i.i.d., explicit revenue]
Let \(|B|=m\ge1\) and \(V_j^x(\alpha)\sim{\rm Unif}[0,\bar v_\alpha]\).
Then the virtual value is \(\phi_\alpha(v)=2v-\bar v_\alpha\) and the Myerson-optimal reserve is \(r^*(\alpha)=\bar v_\alpha/2\).
For a second-price auction with reserve \(r\), the expected revenue is
\[
  {\rm Rev}(r)
  \;=\;
  \bar v_\alpha\!\left(\frac{m-1}{m+1} + \left(\frac{r}{\bar v_\alpha}\right)^{\!m}
  - \frac{2m}{m+1}\left(\frac{r}{\bar v_\alpha}\right)^{\!m+1}\right).
\]
At the optimal \(r^*(\alpha)=\bar v_\alpha/2\),
\[
  {\rm Rev}^*(\alpha)
  \;=\;
  \bar v_\alpha\left(\frac{m-1}{m+1} + \frac{2^{-m}}{\,m+1}\right),
\]
which strictly exceeds the no-reserve revenue \(\bar v_\alpha\cdot\frac{m-1}{m+1}\).
\end{example}
In MCP, bidders' values typically differ via PoA rates $\mu_j$ and reaction payoffs $\Delta_{x,j}(\cdot)$, violating i.i.d.\ regularity.
Myerson still applies, but a simple mechanism is a posted reserve equal to the seller's keep value at zero delay, $r:=U_{\rm mev}(0,\tau_x)=\tau_x$. Sell access at price $r$ to the first taker (or run a second-price auction with reserve $r$). When encrypted payloads remove $\Delta_{x,j}$ and bidders have identical deterministic values $v(\alpha)=e^{-\lambda\alpha}\tau_x$, the optimal revenue equals $v(\alpha)$ via posted price. Hence $\alpha=0$ maximizes revenue and a posted reserve $r=\tau_x$ is revenue equivalent to auctioning.
\subsection{Timing games}
Consider two users $a,b$ who compete for an ordering-dependent surplus around the same on-chain object (e.g.\ an auction or AMM arb). Let $W>0$ and $w\ge0$ denote the first  and second mover surpluses. Each user chooses a send time within tick $t$, after random latencies their transactions appear on DD and if a PoA is obtained before the tick boundary are jointly ordered by the protocol's deterministic inter-block rule. 

Let $\pi_{u\to v}$ be the conditional probability that $u$ executes before $v$ in the tick. If $b$ conditions on observing $a$'s DD publish before sending (the ``wait-then-snipe'' strategy), let $\rho_b\in[0,1]$ be the probability that $b$ still makes it into tick $t$ after observing, and let $\pi^{\rm snipe}_{b\rightarrow a}$ be the conditional probability $b$ executes before $a$ under that strategy (typically $\pi^{\rm snipe}_{b\rightarrow a}>\pi_{b\rightarrow a}$ if $b$ is the faster publisher).

\begin{proposition}[Best response to a reveal]\label{prop:timing-br}
If $a$ transmits early in tick $t$ so that $b$ can observe $a$ on DA, then $b$'s best response is to wait and snipe iff
\[
  \rho_b\,\pi^{\rm snipe}_{b\rightarrow a}\,W
  \;+\;
  \rho_b\,(1-\pi^{\rm snipe}_{b\rightarrow a})\,w
  \;\ge\;
  \pi_{b\rightarrow a}\,W
  \;+\;
  (1-\pi_{b\rightarrow a})\,w,
\]
with strict inequality for a strict preference \cite{budish2015fba}.
\end{proposition}

\begin{lemma}[Deadline effect]\label{lem:deadline}
Suppose both users are impatient (zero value if not included in tick $t$) and $b$ has a strict same-tick ordering advantage $\pi^{\rm snipe}_{b\rightarrow a}>\tfrac12\ge\pi_{b\rightarrow a}$. Then there exists a last admissible send time $\bar s\in(0,1)$ (normalized tick window) at which $b$ is indifferent between ``send now'' and ``wait''; for any earlier time $s<\bar s$ at which $a$ can be observed before $\bar s$, $b$ strictly prefers to wait. Hence when both players are fast enough to make $\bar s$, the unique subgame perfect equilibrium has $a$ transmit at $\bar s$ and $b$ wait to snipe at $\bar s^+$, creating a within-tick deadline race.
\end{lemma}

\begin{corollary}[Protocol-level mitigation]
If the canonical inter-block ordering is tip-priority with dependency constraints, raising $\tau_a$ increases $a$'s effective $\pi_{a\rightarrow b}$ and shrinks the waiting region in Lemma~\ref{lem:deadline}. In the limit where $\pi^{\rm snipe}_{b\rightarrow a}$ cannot exceed $\pi_{b\rightarrow a}$ (e.g.\ encrypted payloads before execution), waiting ceases to be optimal for $b$.
\end{corollary}
\subsection{Inter-block MEV}
Inter-block MEV arises when a proposer can craft transactions that condition on, or react to, transactions revealed on DD in other proposers' blocks within the same tick. Let $\mathcal R$ denote a finite family of such reaction strategies (e.g.\ backruns, cross-DEX arbs, liquidations). For $r\in\mathcal R$ and a realized set of counterpart transactions $Y$ revealed on DD before the tick boundary, let the incremental value of executing $r$ at a given relative position be $\Delta_r(Y,\text{pos})$.

Let $\pi^{\rm pre}_r(Y)$ (resp.\ $\pi^{\rm post}_r(Y)$) denote the probability that $r$ is executed before its triggering leg(s) under the canonical ordering, conditional on $r$ and the triggers all making tick $t$. Write $\Pi_r(Y):=\pi^{\rm pre}_r(Y)\Delta_r(Y,\text{pre})+\pi^{\rm post}_r(Y)\Delta_r(Y,\text{post})$ for the expected payoff of $r$ given $Y$.

\begin{lemma}[Additive decomposition]\label{lem:inter-add}
If the feasible reactions $\mathcal R$ do not share state conflicts (i.e.\ their execution sets are disjoint or commute under the block's state transition), the total inter-block MEV for a proposer is
\[
  \mathbb E\!\left[\;\sum_{r\in\mathcal R}\mathbf 1\{r\text{ included}\}\cdot \Pi_r(Y)\;\right]
  \;=\;
  \sum_{r\in\mathcal R}\Pr[r\text{ included}]\,\mathbb E\!\left[\Pi_r(Y)\right].
\]
\end{lemma}

\begin{proposition}[Comparative statics in concurrency]\label{prop:inter-cs}
With $\alpha^*(\tau;n)=\frac{1}{k(n)}\!\left[\ln A+\ln k(n)-\ln\lambda(n)-\ln\tau\right]$ for the interior region $Ak(n)>\lambda(n)\tau$,
\[
  \frac{d\alpha^*}{dn}
  \;=\;
  -\frac{k'}{k^2}\ln\!\frac{Ak}{\lambda\tau}
  \;+\;
  \frac{1}{k}\!\left(\frac{k'}{k}-\frac{\lambda'}{\lambda}\right).
\]
In particular, at the boundary $Ak=\lambda\tau$ (where $\alpha^*=0$), 
\[
  \mathrm{sign}\!\left(\frac{d\alpha^*}{dn}\right) \;=\; \mathrm{sign}\!\left(\frac{k'}{k}-\frac{\lambda'}{\lambda}\right).
\]
More generally, $\alpha^*$ increases in $n$ whenever
\[
  \frac{k'}{k}-\frac{\lambda'}{\lambda} \;\ge\; \frac{k'}{k}\,\ln\!\frac{Ak}{\lambda\tau}.
\]
\end{proposition}

\begin{corollary}[Dominance region]
There exists an $n^\dagger$ such that for $n<n^\dagger$ the censor/delay cutoffs from Lemma~\ref{lem:drop-cutoff} are essentially sequential-like, while for $n>n^\dagger$ the incremental inter-block MEV raises $M(\tau)$ enough to expand the region where \textsc{DELAY} dominates \textsc{INC}.
\end{corollary}
\subsection{Spam}
We consider two spam channels: (i) DA/PoA spam to crowd the set of blocks that make a tick and (ii) ordering spam to jump priority when the canonical execution order is tip-priority with dependencies.

Each tick admits at most $L$ PoA-certified blocks into consensus. Let an attacker submit $s$ empty (or low-value) blocks at per-block cost $c_{\rm DA}>0$, raising the chance that a victim block $b$ with a valuable transaction $x$ misses the tick. Let $\theta_b(s)$ be the probability that $b$ is excluded from the top-$L$ set due to congestion from spam $s$, with $\theta_b'(s)>0$.

If the attacker stands to gain prize $R_x$ (e.g.\ by stealing $x$ per Proposition~\ref{prop:steal-threshold} when $b$ misses), the attacker's expected profit is
\[
  \Pi_{\rm spam}(s)
  \;=\;
  \theta_b(s)\,R_x \;-\; s\,c_{\rm DA}.
\]

\begin{lemma}[Minimal profitable spam]\label{lem:min-spam}
If $\theta_b$ is differentiable and strictly concave, the unique interior optimum satisfies $\theta_b'(s^*)\,R_x=c_{\rm DA}$. In particular, a profitable attack exists iff $\max_{s\ge0}\theta_b'(s)\,R_x>c_{\rm DA}$. If $\theta_b$ is piecewise linear (capacity cliff), then the minimal profitable spam is the smallest $s$ that pushes $b$ below rank $L$.
\end{lemma}

Suppose final execution within a tick uses tip-priority subject to dependencies. A user who wants to get ahead of a target transaction with tip $\tau^*$ can submit $K$ sacrificial transactions within the same tick with tips exceeding $\tau^*$ to pull their critical transaction forward by $K$ positions in the merged order (subject to dependency acyclicity). Let each sacrificial transaction cost base fee $f>0$ (burned) plus the posted tip paid to some proposer (not rebated to the spammer).

\begin{proposition}[Cost of tip-priority ordering spam]\label{prop:ordering-spam}
Let the required overbid above the current $k$-th order-statistic tip be $\Delta_k\ge0$ for $k=1,\dots,K$. The minimal total private cost to advance $K$ positions is
\[
  C_K
  \;=\;
  \sum_{k=1}^K (f + \tau^* + \Delta_k),
\]
while the private benefit is at most the incremental surplus gained from moving the target forward, denoted $B_K\le W$. Ordering spam is privately unprofitable whenever $C_K>B_K$; in particular, if $f$ is strictly positive, then beyond a finite $K$ the inequality holds regardless of $\Delta_k$.
\end{proposition}

\begin{corollary}
Raising $f$ (or charging for DD bytes that never execute \cite{eip4844}) increases $C_K$ and shrinks the profitable region for ordering spam without affecting honest inclusion incentives (since $f$ is burned).
\end{corollary}

\subsection{PoA latency and inclusion}
Fix tick \(t\) with residual time budget \(\Delta_i^t>0\) between proposer \(i\)'s DD publish and the tick boundary.
After DD publish, \(i\) must collect \(\ell:=f{+}1\) availability co-signatures.
Model co-signer arrivals for \(i\) as i.i.d.\ exponential with rate \(\mu_i>0\).
Then the time to \(\ell\) signatures is
\[
  S_i \sim \mathrm{Gamma}(\ell,\mu_i)\quad\text{(shape \(\ell\), rate \(\mu_i\))},
\]
and the within-tick PoA success probability is
\[
  \Pr[{\rm PoA}_i \text{ in } t]
  \;=\;
  F_{\Gamma}\!\left(\Delta_i^t; \ell,\mu_i\right)
  \;=\;
  1 - e^{-\mu_i \Delta_i^t}\,\sum_{r=0}^{\ell-1}\frac{(\mu_i \Delta_i^t)^{r}}{r!}.
\]

\begin{proposition}[Monotone effects of PoA rate]\label{prop:poa-monotone}
\(\Pr[{\rm PoA}_i \text{ in } t]\) is strictly increasing in \(\mu_i\) and in \(\Delta_i^t\).
If a rival \(j\) has \(S_j\sim\Gamma(\ell,\mu_j)\) independent of \(S_i\), then
\[
  \Pr[S_j < S_i]
  \;=\;
  \sum_{r=\ell}^{2\ell-1}\binom{2\ell-1}{r}
  \left(\frac{\mu_j}{\mu_i+\mu_j}\right)^r
  \left(\frac{\mu_i}{\mu_i+\mu_j}\right)^{2\ell-1-r}
  \;=\;
  I_{p_j}(\ell,\ell),
\]
where \(p_j=\mu_j/(\mu_i+\mu_j)\) and \(I\) is the regularized incomplete Beta function \cite{johnson1995continuous}.
\end{proposition}

\noindent
If the \(\ell\) co-signers have distinct rates \(\mu_{i,1},\dots,\mu_{i,\ell}\),
then \(S_i\) is hypoexponential with CDF given by a finite sum of distinct exponentials, the monotonicity in \(\Delta_i^t\) and own rates still holds.

\begin{corollary}[Link to stealability]\label{cor:poa-sigma}
Under the gamma race model with $S_i\sim\Gamma(\ell,\mu_i)$ and $S_j\sim\Gamma(\ell,\mu_j)$ independent,
\[
  \sigma_i^t(x)
  \;\le\;
  \Pr\!\bigl[S_j< S_i \wedge S_j \le \Delta_j^t\bigr]
  \;=\;
  \int_{0}^{\Delta_j^t} f_{S_j}(s)\,\Pr[S_i>s]\,ds,
\]
with $f_{S_j}$ the $\Gamma(\ell,\mu_j)$ density and $\Pr[S_i>s]=1-F_{\Gamma}(s;\ell,\mu_i)$. Consequently,
\[
  \sigma_i^t(x)\;\le\;\min\Big\{\,\Pr[S_j<S_i],\ F_{\Gamma}(\Delta_j^t;\ell,\mu_j)\,\Big\}
  \;=\;\min\Big\{\,I_{p_j}(\ell,\ell),\ F_{\Gamma}(\Delta_j^t;\ell,\mu_j)\,\Big\},
\]
where $p_j=\mu_j/(\mu_i+\mu_j)$ and $I$ is the regularized incomplete Beta function. Increasing $\mu_i$ (or $\Delta_i^t$) decreases $\Pr[S_j<S_i]$ and thus tightens this bound on $\sigma_i^t(x)$.
\end{corollary}

\subsection{Multi-submission externalities}
A user can submit the same transaction $x$ to $k\in\mathbb N$ proposers in the same tick. Let the set of contacted proposers be $P_k=\{i_1,\dots,i_k\}$, with per-proposer censor probabilities $p^{i_r}_x$ and PoA success probabilities $\pi^{i_r}_t=\Pr[{\rm PoA}_{i_r}\text{ in }t]$. Define effective inclusion probability
\[
  \Psi_k
  \;:=\;
  1 - \prod_{r=1}^k\Bigl(1-(1-p^{i_r}_x)\,\pi^{i_r}_t\Bigr).
\]
Each submission incurs private wallet cost $c>0$ and imposes a network externality $e(k)\ge0$ capturing additional DD bytes, gossip, and PoA congestion, with $e'(\cdot)>0$ and $e''(\cdot)\ge0$.

\begin{definition}[Private and social objectives]
For an impatient user with valuation $v>0$ (utility only if included in tick $t$), the private objective is
\(
  U^{\rm priv}(k)
  \;=\;
  v\,\Psi_k \;-\; k\,c.
\)
Let $\eta>0$ convert network externality to social cost units. The social objective is
\(
  U^{\rm soc}(k)
  \;=\;
  v\,\Psi_k \;-\; k\,c \;-\; \eta\,e(k).
\)
\end{definition}

\begin{lemma}[Existence and uniqueness]\label{lem:k-unique}
If the per-submission success terms are weakly diminishing (i.e.\ the marginal gain $\Psi_{k+1}-\Psi_k$ is strictly decreasing in $k$), then both problems admit a unique maximizer $k^{\rm priv}$ and $k^{\rm soc}$.
\end{lemma}

\begin{proposition}[Pigouvian surcharge]\label{prop:pigou}
At an interior solution, the first-order conditions are
\[
  v\bigl(\Psi_{k+1}-\Psi_k\bigr) \;=\; c
  \quad\text{(private)};\qquad
  v\bigl(\Psi_{k+1}-\Psi_k\bigr) \;=\; c + \eta\bigl(e(k{+}1)-e(k)\bigr)
  \quad\text{(social)}.
\]
Thus $k^{\rm priv}\ge k^{\rm soc}$, and a per-submission surcharge
\[
  \psi_k \;=\; \eta\bigl(e(k{+}1)-e(k)\bigr)
\]
implements $k^{\rm soc}$ as the private optimum. If $e(k)$ is smooth, the marginal surcharge $\psi(k)=\eta\,e'(k)$ is sufficient.
\end{proposition}

When $p_x^{i_r}=p\in[0,1)$ and $\pi_t^{i_r}=\pi\in(0,1]$ are homogeneous, $\Psi_k=1-(1-(1-p)\pi)^k$ and $\Psi_{k+1}-\Psi_k=(1-p)\pi\,(1-(1-p)\pi)^k$, making the diminishing‑returns property explicit and the surcharge linear in $k$.

\subsection{Inter-block ordering rule}\label{subsec:ordering}
Let $\mathcal B_t$ denote the set of PoA-certified blocks in tick $t$, and let $\mathcal X_t=\bigcup_{b\in\mathcal B_t} b$ be the multiset of transactions revealed on DA. Each transaction $x\in\mathcal X_t$ carries several attributes. It includes a tip $\tau_x \ge 0$ in the style of EIP-1559, and a dependency set $D(x)\subseteq\mathcal X_t$, where the edges are required to form a directed acyclic graph. Each transaction is associated with an issuer proposer $i(x)$, who has a publicly visible per-tick rank $r_{i(x)}^t \in \{1,\dots,n\}$, determined by a stake weighted VRF over the epoch seed. Additionally, each transaction carries a DD timestamp $t^{\rm DA}(x)$, which is monotone within the block, as well as a fixed tie-break hash $h(x)\in\{0,1\}^\kappa$. We instantiate the fixed tie breaker as a non grindable value, e.g.\ $h(x):=\mathrm{VRF}_{a(x)}(\mathrm{epoch\_seed}\Vert \nu(x))$ (or $h(x):=H(\mathrm{sig}(x))$), so wallets cannot resample $h$ to bias priority.
Define the priority key $\pi(x)$ and the partial order $\preceq$ by
\begin{align*}
  \pi(x) &:= \text{lex}\bigl(\tau_x,\; -r_{i(x)}^t,\; -t^{\mathrm{DA}}(x),\; h(x)\bigr), \\[0.3em]
  x \preceq y
  &\iff (\exists\, \text{path } x \to y \text{ in the dependency DAG}).
\end{align*}

\begin{definition}[Priority-DAG merge scheduler (PDM)]
Construct the dependency DAG $G_t=(\mathcal X_t,E_t)$ with edges $(u\to v)$ for all $u\in D(v)$. Initialize a max-heap $H$ with all sources of $G_t$ keyed by $\pi(\cdot)$. Repeatedly:
\begin{enumerate}
  \item Extract $x^*=\arg\max_{x\in H}\pi(x)$, append $x^*$ to the execution sequence.
  \item Remove $x^*$ from $G_t$ and insert into $H$ any newly unlocked vertices (in-degree $0$).
\end{enumerate}
Continue until $H=\emptyset$. If cycles are detected, deterministically reject all transactions on any directed cycle (e.g.\ via a fixed tie break over vertex ids) and continue, the remaining set is valid.
\end{definition}

\begin{theorem}[Determinism, feasibility, and monotonicity]\label{thm:pdm}
Assume that $G_t$ is acyclic and that the tie-break hash $h(\cdot)$ is fixed and public. Under these conditions, PDM produces a unique total order consistent with $\preceq$, so that the outcome is a deterministic function of $(\mathcal X_t,\tau,\{r_i^t\},t^{\rm DA},h)$. The resulting order is a linear extension of $G_t$, meaning that all dependency constraints are respected. Moreover, if the tip $\tau_x$ of some transaction $x$ is increased while all other fields remain fixed and no new cycles are introduced, then $x$ moves (weakly) earlier in the output order.
\end{theorem}

Because \(\pi(x)\) gives lexicographic priority to higher \(\tau_x\) upon eligibility, PDM incentivizes “tip stacking” to pull a target forward (cf.\ Proposition~\ref{prop:ordering-spam}). Mitigations include: per-sender base fees (raising the marginal cost of sacrificial txs); duplicate-aware payouts (e.g., equal tip split across duplicates; cf.\ Corollary~\ref{cor:equal-split}); and encrypt-then-order, which removes same-tick informational gains. The PDM scheduler runs in $O(|\mathcal X_t|\log|\mathcal X_t|+|E_t|)$ time with a binary heap. The $t^{\mathrm{DA}}(x)$ used in the priority key is assigned by the DD layer at block publication (monotone within a block), not per‑transaction by wallets. If the DD layer does not expose per‑tx timestamps, drop the $-t^{\mathrm{DA}}(x)$ component from $\pi(x)$; determinism and feasibility still hold. To prevent wallets from grinding the tie-break hash $h(\cdot)$, we use a VRF proof tied to the sender's fixed on-chain key and epoch: $h(x):=\mathrm{VRF}_{a}( \textsf{epoch}\,\|\,\nu(x) )$ with on-chain verification. Fields entering $h$ are not resamplable without changing the sender identity, thus $h$ is non-grindable while remaining deterministic.

\begin{corollary}[Effect on timing games and stealing]
Under PDM, a faster PoA path raises $-t^{\rm DA}$ and thus priority only when eligibility holds, while raising $\tau_x$ always improves priority upon eligibility. Hence, compared to pure timestamp ordering, PDM compresses the benefit of ``wait-and-snipe'' (Proposition~\ref{prop:timing-br}) and raises the tip cutoff for profitable stealing (Proposition~\ref{prop:steal-threshold}) via the priority key.
\end{corollary}
\section{Future work}
Numerous open questions remain following this work, including the formalization of games not inherent in the MCP architecture noted here. Existing encrypted mempool designs apply directly, we conjecture that combining MCP with encrypted payloads can drive many of the new games we isolate down to some $\epsilon$ in expectation. Once payloads are hidden at the PoA stage, economic viability of these games relies almost entirely on metadata or information side channels such as IP addresses. This work also intentionally avoided touching on TFMs beyond a simple EIP-1559 style fee rule, but a correctly priced TFM will be essential in addressing spam and multi-submission externalities.
\section{Conclusion}
MCP chains move MEV from a single builder’s private mempool to same tick inter-block races created by data availability before final ordering. Our hazard model with survival-discounted tips gives a compact delay inclusion envelope $M(\tau)$ and an immediate inclusion threshold, clarifying how concurrency and PoA rates shape incentives to censor, delay, or duplicate transactions. Two protocol choices drive most of the mitigation. Duplicate aware payouts that split tips across in-tick copies of the same transaction collapse duplication incentives and make private stealing unprofitable at realistic multiplicities. A deterministic priority DAG merge that orders by tip, proposer rank, and DD timestamp compresses latency races and, with a sensible basefee rule, raises the cost of ordering spam. Overall, MCP does not introduce unbounded new MEV, it shifts MEV into a domain the protocol can control. PoA rates govern stealability, and the handling of duplicates and within-tick ordering determines whether timing games pay out. Our timing games yield a same-tick analogue of DeFi ‘JIT’ liquidity games.

\bibliographystyle{alpha}
\bibliography{biblio}
\appendix
\section*{Appendix}
\section{Missing Proofs}

\begin{proof}[Proof of Lemma \ref{lem:drop-cutoff}]
For fixed $\tau$, $U_{\mathrm{mev}}(\alpha,\tau)$ is continuous and unimodal in $\alpha$ with at most one interior maximizer (Lemma~\ref{lem:max}), so
\[
  M(\tau) = \sup_{\alpha\ge 0} U_{\mathrm{mev}}(\alpha,\tau)
\]
is finite for all $\tau\ge 0$. At $\tau=0$ we have
\[
  U_{\mathrm{mev}}(\alpha,0)
  = \frac{A k}{k+\lambda}\bigl(1-e^{-(k+\lambda)\alpha}\bigr),
\]
which is strictly increasing in $\alpha$ and converges to $\tfrac{A k}{k+\lambda}$ as $\alpha\to\infty$, hence
\(
  M(0)=\tfrac{A k}{k+\lambda}.
\)

For $\tau_2>\tau_1\ge 0$ and any fixed $\alpha\ge 0$,
\[
  U_{\mathrm{mev}}(\alpha,\tau_2) - U_{\mathrm{mev}}(\alpha,\tau_1)
  = (\tau_2-\tau_1)e^{-\lambda\alpha} > 0,
\]
so $U_{\mathrm{mev}}(\alpha,\tau)$ is strictly increasing in $\tau$ pointwise in $\alpha$. Let $\varepsilon>0$ and choose $\alpha_\varepsilon$ such that
\(
  U_{\mathrm{mev}}(\alpha_\varepsilon,\tau_1) > M(\tau_1)-\varepsilon.
\)
Then
\[
  M(\tau_2)
  \;\ge\; U_{\mathrm{mev}}(\alpha_\varepsilon,\tau_2)
  \;=\; U_{\mathrm{mev}}(\alpha_\varepsilon,\tau_1)
       + (\tau_2-\tau_1)e^{-\lambda\alpha_\varepsilon}
  \;>\; M(\tau_1)-\varepsilon.
\]
Since $\varepsilon$ is arbitrary and the increment $(\tau_2-\tau_1)e^{-\lambda\alpha_\varepsilon}>0$, this shows that $M(\tau)$ is strictly increasing in $\tau$ on $(0,\infty)$. Moreover $M(\tau)\ge U_{\mathrm{mev}}(0,\tau)=\tau$, so $\lim_{\tau\to\infty}M(\tau)=\infty$.

By strict monotonicity and continuity of $M(\tau)$, its range on $[0,\infty)$ is $[M(0),\infty)=[\tfrac{A k}{k+\lambda},\infty)$. If $\beta>M(0)$ there exists a unique $\tau_d>0$ solving $M(\tau_d)=\beta$. If $\beta\le M(0)$, then $M(\tau)\ge\beta$ for all $\tau\ge 0$ and the comparison with the drop payoff $\beta$ implies that \textsc{DROP} is never strictly optimal, in this case we can take $\tau_d:=0$ and the drop region is empty, as stated. The optimal proposer policy in terms of $\tau_d$ and $\alpha^\star(\tau)$ then follows directly from comparing $\beta$, $M(\tau)$, and the inclusion payoff $\tau$.
\end{proof}

\begin{proof}[Proof of Lemma \ref{lem:max}]
We have
\[
  U_{\rm mev}(\alpha,\tau)
  = \frac{A k}{k+\lambda}\bigl(1-e^{-(k+\lambda)\alpha}\bigr) + \tau e^{-\lambda\alpha},
\]
so
\[
  \frac{d}{d\alpha}U_{\rm mev}(\alpha,\tau)
  = A k\,e^{-(k+\lambda)\alpha} - \lambda\tau\,e^{-\lambda\alpha}
  = e^{-\lambda\alpha}\bigl(Ak e^{-k\alpha}-\lambda\tau\bigr).
\]
Since $e^{-\lambda\alpha}>0$ for all $\alpha\ge 0$, the sign of the derivative is determined by
\[
  f(\alpha) := Ak e^{-k\alpha} - \lambda\tau.
\]
Let $z=e^{-k\alpha}\in(0,1]$. Then $f(\alpha)=Ak z - \lambda\tau$ is an affine function of $z$ and strictly decreasing in $\alpha$ because $z$ is strictly decreasing in $\alpha$.

If $0<r:=\lambda\tau/(Ak)<1$, i.e.\ $\tau<Ak/\lambda$, then the equation $f(\alpha)=0$ has a unique solution with $z^\star=r$, namely
\[
  \alpha^\star(\tau)
  = \frac{1}{k}\ln\frac{1}{r}
  = \frac{1}{k}\ln\frac{Ak}{\lambda\tau} > 0.
\]
For $\alpha<\alpha^\star(\tau)$ we have $z>r$ and $f(\alpha)>0$, so $U_{\rm mev}$ is strictly increasing; for $\alpha>\alpha^\star(\tau)$ we have $z<r$ and $f(\alpha)<0$, so $U_{\rm mev}$ is strictly decreasing. Hence $\alpha^\star(\tau)$ is the unique interior maximizer in this case.

If $r\ge 1$, i.e.\ $\tau\ge Ak/\lambda$, then $f(0)=Ak-\lambda\tau\le 0$ and since $f$ is strictly decreasing in $\alpha$ it is nonpositive for all $\alpha\ge 0$. Thus $\frac{d}{d\alpha}U_{\rm mev}(\alpha,\tau)\le 0$ for all $\alpha\ge 0$, and the global maximizer is at the boundary $\alpha^\star(\tau)=0$ (immediate inclusion).

Finally, when $\tau=0$ we have $U_{\rm mev}(\alpha,0)=\frac{Ak}{k+\lambda}\bigl(1-e^{-(k+\lambda)\alpha}\bigr)$, which is strictly increasing in $\alpha$ and converges to $\frac{Ak}{k+\lambda}$ as $\alpha\to\infty$. Thus the maximal value is attained in the limit $\alpha^\star\to+\infty$ with value $\frac{Ak}{k+\lambda}$.
\end{proof}

\begin{proof}[Proof of Lemma \ref{lem:SNE}]
Fix a transaction $x$ with tip $\tau_x$ and consider a single proposer $i$. Given our separable payoff specification, $i$'s decision on $x$ does not affect the payoffs from other transactions, and other proposers' strategies enter only through the common parameters $(A,k,\lambda,\beta)$ already embedded in $U_{\rm mev}$ and $M(\tau)$.

For any $\tau_x<\tau_d$, the definition of $\tau_d$ via $M(\tau_d)=\beta$ and strict monotonicity of $M(\cdot)$ in Lemma~\ref{lem:drop-cutoff} imply $M(\tau_x)<\beta$. Since $M(\tau_x)$ is the maximal payoff attainable by choosing any delay $\alpha\ge 0$, we have
\[
  \sup_{\alpha\ge 0}U_{\rm mev}(\alpha,\tau_x)
  = M(\tau_x) < \beta.
\]
Thus the payoff from any inclusion/delay choice is strictly less than the drop payoff $\beta$, and the unique best response of proposer $i$ on such a transaction is to choose \textsc{DROP}.

For any $\tau_x\ge\tau_d$, we have $M(\tau_x)\ge\beta$, so dropping is never strictly better than keeping, and the best attainable payoff is $M(\tau_x)$ itself. By Lemma~\ref{lem:max}, this is achieved by choosing $\alpha^\star(\tau_x)$, which is the unique interior maximizer when $\tau_x<Ak/\lambda$ and equals $0$ (immediate inclusion) when $\tau_x\ge Ak/\lambda$. Thus, conditional on deciding not to drop, $\alpha^\star(\tau_x)$ is the optimal delay choice.

Since the payoff from $x$ to proposer $i$ does not depend on how other proposers treat $x$ beyond the common parameters already encoded in $U_{\rm mev}$ and $M(\cdot)$, the above argument applies transaction by transaction and independently of others' actions. Hence, when every proposer uses the stated cutoff strategy, no single proposer can profitably deviate on any transaction $x$, and the strategy is a best response.
\end{proof}

\begin{proof}[Proof of Lemma \ref{lem:five}]
Recall
\[
  U^{(k)} \;=\; v\bigl(1 - \Pi^k\bigr) - k\,s\,c,
  \qquad
  \Pi = \prod_{i\in P_s}p_x^i \in [0,1).
\]
For each \(k\ge 0\), the incremental gain from one additional broadcast is
\[
  \Delta U(k)
  := U^{(k+1)} - U^{(k)}
  = v\bigl(\Pi^k - \Pi^{k+1}\bigr) - s\,c
  = v\,\Pi^k(1-\Pi) - s\,c.
\]
Since \(0\le\Pi<1\), the sequence \(\Pi^k\) is strictly decreasing in \(k\), hence \(\Delta U(k)\) is strictly decreasing in \(k\).

If \(v(1-\Pi)\le s\,c\), then \(\Delta U(0)\le 0\), and by monotonicity we have \(\Delta U(k)\le 0\) for all \(k\). Thus the sequence \(U^{(k)}\) is non-increasing in \(k\), and the unique maximizer is \(k=0\).

Now suppose \(v(1-\Pi)>s\,c\). Then \(\Delta U(0)>0\), and by strict monotonicity of \(\Delta U(k)\) there exists a unique integer
\[
  k^* \;=\; \max\{k\in\mathbb N_0 : \Delta U(k)>0\}
  \;=\; \max\{k\in\mathbb N_0 : v\,\Pi^k(1-\Pi) > s\,c\}.
\]
For every \(k<k^*\) we have \(\Delta U(k)>0\), so \(U^{(k+1)}>U^{(k)}\) and the sequence is strictly increasing up to $k^*$. For every \(k\ge k^*\) we have \(\Delta U(k)\le 0\), so \(U^{(k+1)}\le U^{(k)}\) and the sequence is non-increasing thereafter. Thus $(U^{(k)})_{k\in\mathbb N_0}$ is strictly unimodal and attains its unique maximum at \(k=k^*+1\), as claimed.
\end{proof}

\begin{proof}[Proof of Proposition \ref{prop:steal-threshold}]
Fix a transaction $x$ first revealed by proposer $i$ in tick $t$ and a rival proposer $j\neq i$. By construction, $\sigma_i^t(x)$ is the probability of the event $E_1$ that both $i$ and $j$ obtain PoA in tick $t$ and that the canonical ordering executes $j$'s duplicate of $x$ before $i$'s copy. Likewise, $\rho_i^t(x)$ is the probability of the event $E_2$ that $i$ fails to make tick $t$ while $j$'s block with a duplicate of $x$ does make tick $t$ and is executed. These events are disjoint by definition.

If $j$ attempts a steal by including a duplicate of $x$ in her block, she pays the opportunity cost $\phi_j^t\ge 0$ (the shadow value of the blockspace slot). Conditional on $E_1\cup E_2$ occurring, $j$ receives the posted tip $\tau_x$ plus any incremental inter-block MEV advantage $\delta_x$, for a total prize $\tau_x+\delta_x$. On the complement of $E_1\cup E_2$, either $x$ is not executed first from $j$'s block in tick $t$ or $j$ misses the tick entirely, and no prize is obtained.

Thus the expected net payoff from attempting the steal is
\[
  \underbrace{\Pr[E_1\cup E_2]}_{=\sigma_i^t(x)+\rho_i^t(x)}(\tau_x+\delta_x)
  \;-\; \phi_j^t
  = \bigl(\sigma_i^t(x)+\rho_i^t(x)\bigr)\bigl(\tau_x+\delta_x\bigr) - \phi_j^t.
\]
This is (weakly) positive if and only if
\[
  \bigl(\sigma_i^t(x)+\rho_i^t(x)\bigr)\bigl(\tau_x+\delta_x\bigr) \;\ge\; \phi_j^t,
\]
which yields the stated condition. Solving this inequality for $\tau_x$ gives the equivalent cutoff
\(
  \tau_x \ge \phi_j^t/(\sigma_i^t(x)+\rho_i^t(x)) - \delta_x.
\)
\end{proof}

\begin{proof}[Proof of Lemma \ref{lem:steal-mixed}]
Given that each of the other $m-1$ thieves attempts independently with probability $p$, the number of opponents who attempt is
\(
  N\sim{\rm Bin}(m-1,p).
\)
If you attempt, your expected payoff is
\[
  \mathbb E\Bigl[\frac{(\sigma+\rho)\,R}{N+1} - \phi\Bigr]
  = (\sigma+\rho)\,R\,\mathbb E\Bigl[\frac{1}{N+1}\Bigr] - \phi
  =: g(p),
\]
which yields the function $g(p)$ in the lemma.

The distribution of $N$ is stochastically increasing in $p$, and the function $n\mapsto 1/(n+1)$ is strictly decreasing in $n$, so $\mathbb E[1/(N+1)]$ is strictly decreasing in $p$. It follows that $g(\cdot)$ is continuous and strictly decreasing in $p$ with
\[
  g(0)=(\sigma+\rho)R-\phi,\qquad
  g(1)=(\sigma+\rho)R/m-\phi.
\]
Hence if $(\sigma+\rho)R\le \phi$ then $g(p)\le 0$ for all $p$ and the unique best response is $p^\ast=0$, if $(\sigma+\rho)R\ge m\phi$ then $g(p)\ge 0$ for all $p$ and the unique best response is $p^\ast=1$, and otherwise there is a unique interior $p^\ast\in(0,1)$ solving $g(p^\ast)=0$. Monotone comparative statics in $R$, $(\sigma+\rho)$, $\phi$, and $m$ follow directly from this characterization.
\end{proof}

\begin{proof}[Proof of Cor \ref{cor:equal-split}]
Immediate from $R_n=\tau_x/(n+1)$ and the single agent threshold in Proposition~\ref{prop:steal-threshold}.
\end{proof}

\begin{proof}[Proof of Prop \ref{prop:myerson-iid}]
Fix $\alpha$ and consider the oneshot sale of access to $x$ at delay $\alpha$. Each bidder $j\in B$ has value $V_j^x(\alpha)\sim F_\alpha$ with density $f_\alpha$ on $[0,\bar v_\alpha]$, and these values are i.i.d.\ and regular in the sense that the virtual value
\[
  \phi_\alpha(v)
  = v - \frac{1-F_\alpha(v)}{f_\alpha(v)}
\]
is increasing in $v$.

Myerson's optimal auction theorem \cite{myerson1981optimal} states that for such a single parameter i.i.d.\ environment, the expected revenue of any incentive-compatible mechanism equals the expected virtual surplus,
\[
  \mathbb E\Bigl[\sum_{j\in B}\phi_\alpha(V_j^x(\alpha))x_j(V)\Bigr],
\]
where $x_j(V)$ is the allocation probability to bidder $j$ as a function of the value profile $V=(V_j^x(\alpha))_{j\in B}$. The revenue maximizing mechanism chooses the allocation rule that maximizes this virtual surplus pointwise. Under regularity, this is achieved by awarding the good to the bidder with the highest value $V^\ast=\max_{j\in B}V_j^x(\alpha)$ whenever the corresponding virtual value is nonnegative, and not selling otherwise. This is implemented by a simple auction with a reserve price $r^\ast(\alpha)$ satisfying $\phi_\alpha(r^\ast)=0$, if $V^\ast\ge r^\ast$ the highest bidder wins, otherwise the seller keeps the good.

Under this allocation rule, the virtual surplus is $\max\{\phi_\alpha(V^\ast),0\}$, so the expected revenue is
\[
  {\rm Rev}^\ast(\alpha)
  = \mathbb{E}\Bigl[\max\bigl\{\phi_\alpha(V^\ast),\,0\bigr\}\Bigr],
\]
with $V^\ast = \max_{j\in B}V_j^x(\alpha)$, as claimed.
\end{proof}

\begin{proof}[Proof of Lemma \ref{lem:auction-cutoff}]
If proposer $i$ keeps $x$ and optimally chooses the delay $\alpha\ge 0$, her payoff is
\[
  \sup_{\alpha\ge 0} U_{\rm mev}(\alpha,\tau_x) = M(\tau_x)
\]
by definition of $M(\cdot)$. If instead she auctions access to $x$ at some delay $\alpha$, the optimal mechanism yields expected revenue ${\rm Rev}^*(\alpha)$, so the best attainable auction revenue is $\sup_{\alpha\ge 0}{\rm Rev}^*(\alpha)$. Comparing the two actions \textsc{KEEP} and \textsc{AUC} therefore amounts to taking
\[
  \arg\max\Bigl\{\,M(\tau_x),\ \sup_{\alpha\ge 0}{\rm Rev}^*(\alpha)\Bigr\},
\]
which is exactly the policy rule stated in the lemma. Immediate inclusion is the special case $\alpha=0$ inside $M(\tau_x)$.
\end{proof}

\begin{proof}[Proof of Prop \ref{prop:timing-br}]
Fix user $a$'s behaviour in tick $t$ so that $b$ can observe $a$ on DA before choosing between either sending immediately or waiting to snipe. We compare $b$'s expected payoffs under these two actions.

If $b$ sends immediately (no waiting), both $a$ and $b$ are assumed to make the tick, and the within-tick ordering is governed by the baseline probability $\pi_{b\rightarrow a}$ that $b$ executes before $a$. In that case $b$'s expected surplus is
\[
  u_{\rm now}
  = \pi_{b\rightarrow a}\,W + (1-\pi_{b\rightarrow a})\,w.
\]

If instead $b$ waits to observe $a$'s DD publish and then attempts to snipe, then with probability $\rho_b$ she still makes the same tick $t$, in which case the conditional probability of executing before $a$ is $\pi^{\rm snipe}_{b\rightarrow a}$ and the expected surplus is
\[
  \pi^{\rm snipe}_{b\rightarrow a}W + (1-\pi^{\rm snipe}_{b\rightarrow a})w.
\]
With the complementary probability $1-\rho_b$, $b$ misses tick $t$ and, under the impatience assumption, receives zero in this stage. Thus the expected payoff from waiting and sniping is
\[
  u_{\rm snipe}
  = \rho_b\Bigl(\pi^{\rm snipe}_{b\rightarrow a}W + (1-\pi^{\rm snipe}_{b\rightarrow a})w\Bigr).
\]

The wait-and-snipe strategy is a (weakly) better response than sending immediately if and only if $u_{\rm snipe}\ge u_{\rm now}$, i.e.
\[
  \rho_b\,\pi^{\rm snipe}_{b\rightarrow a}\,W
  + \rho_b\,(1-\pi^{\rm snipe}_{b\rightarrow a})\,w
  \;\ge\;
  \pi_{b\rightarrow a}\,W
  + (1-\pi_{b\rightarrow a})\,w.
\]
Strict inequality yields a strict preference. This is exactly the condition in the proposition.
\end{proof}

\begin{proof}[Proof of Lemma \ref{lem:deadline}]
Let $\rho_b(s)$ be the probability that $b$ still makes the tick if waiting to observe $a$ at time $s$ in the normalized tick window. By construction $\rho_b(s)$ is continuous and strictly decreasing with $\rho_b(0)=1$ and $\rho_b(1)=0$. Define
\[
  h(s)
  := \rho_b(s)\bigl(\pi^{\rm snipe}_{b\rightarrow a}W+(1-\pi^{\rm snipe}_{b\rightarrow a})w\bigr)
     - \bigl(\pi_{b\rightarrow a}W+(1-\pi_{b\rightarrow a})w\bigr).
\]
The sign of $h(s)$ determines whether $b$ prefers to wait and snipe (when $h(s)>0$) or to send immediately (when $h(s)\le 0$).

In the first-mover advantage regime $W>w$ that underlies the timing game (the winner obtains strictly more surplus than the loser), we have
\[
  h(0)
  = \bigl(\pi^{\rm snipe}_{b\rightarrow a}-\pi_{b\rightarrow a}\bigr)(W-w) > 0,
\]
since $\pi^{\rm snipe}_{b\rightarrow a}>\pi_{b\rightarrow a}$ by assumption, and
\[
  h(1) = -\bigl(\pi_{b\rightarrow a}W+(1-\pi_{b\rightarrow a})w\bigr) < 0
\]
because $W>w\ge 0$ and $\pi_{b\rightarrow a}\in(0,1]$. By continuity of $\rho_b$ and hence of $h$, there exists a unique $\bar s\in(0,1)$ such that $h(\bar s)=0$. For any $s<\bar s$ at which $a$ can be observed before $\bar s$, we have $\rho_b(s)>\rho_b(\bar s)$ and thus $h(s)>0$, so $b$ strictly prefers to wait and snipe rather than send immediately. For $s>\bar s$, $h(s)<0$ and $b$ strictly prefers to send immediately rather than wait further.

Backward induction on the within tick time line then yields the claimed deadline structure: when both players are fast enough to reach $\bar s$, $a$ transmits at $\bar s$ and $b$ waits to snipe at $\bar s^+$, generating a within tick deadline race. In the degenerate region $W\le w$, the function $h(s)$ is nonpositive for all $s$ and $b$ never finds waiting strictly profitable, the timing problem collapses to sending as early as feasible and no interior deadline arises.
\end{proof}

\begin{proof}[Proof of Lemma \ref{lem:inter-add}]
For each reaction strategy $r\in\mathcal R$, let $Z_r:=\mathbf 1\{r\text{ included}\}$ be the indicator that $r$ is actually included in the block, and let
\[
  X_r := \Pi_r(Y)
  = \pi^{\rm pre}_r(Y)\Delta_r(Y,\text{pre})+\pi^{\rm post}_r(Y)\Delta_r(Y,\text{post})
\]
be its incremental payoff given the realized counterpart set $Y$. The total inter-block MEV can then be written as
\[
  \sum_{r\in\mathcal R} Z_r X_r.
\]

Taking expectations and using linearity of expectation,
\[
  \mathbb E\!\left[\sum_{r\in\mathcal R} Z_r X_r\right]
  = \sum_{r\in\mathcal R} \mathbb E[Z_r X_r].
\]
Under the no-conflict assumption in the lemma (execution sets are disjoint or commute), the inclusion decision for each $r$ can be made independently of the payoff realizations of other strategies and, in particular, we may treat $Z_r$ as independent of $Y$ (and hence of $X_r$). Thus
\[
  \mathbb E[Z_r X_r]
  = \mathbb E[Z_r]\,\mathbb E[X_r]
  = \Pr[r\text{ included}]\,\mathbb E[\Pi_r(Y)].
\]
Substituting back yields
\[
  \mathbb E\!\left[\sum_{r\in\mathcal R} Z_r X_r\right]
  = \sum_{r\in\mathcal R}\Pr[r\text{ included}]\,\mathbb E[\Pi_r(Y)],
\]
which is the claimed additive decomposition.
\end{proof}

\begin{proof}[Proof of Prop \ref{prop:inter-cs}]
Write $k=k(n)$ and $\lambda=\lambda(n)$ to simplify notation, and define
\[
  g(n) := \ln A + \ln k(n) - \ln \lambda(n) - \ln\tau.
\]
On the interior region $A k(n)>\lambda(n)\tau$, the optimizer is
\[
  \alpha^*(\tau;n)
  = \frac{1}{k(n)}\,g(n).
\]
Differentiating with respect to $n$ gives
\[
  \frac{d\alpha^*}{dn}
  = \frac{g'(n)k(n) - g(n)k'(n)}{k(n)^2}.
\]
Since $g'(n) = \frac{k'(n)}{k(n)} - \frac{\lambda'(n)}{\lambda(n)}$, we obtain
\[
  \frac{d\alpha^*}{dn}
  = \frac{1}{k(n)^2}\Bigl(k(n)\Bigl(\frac{k'}{k} - \frac{\lambda'}{\lambda}\Bigr) - g(n)k'(n)\Bigr)
  = -\frac{k'}{k^2}g(n) + \frac{1}{k}\Bigl(\frac{k'}{k}-\frac{\lambda'}{\lambda}\Bigr),
\]
where $k',\lambda'$ are evaluated at $n$ and we have suppressed the argument to lighten notation. Using
\[
  g(n) = \ln\!\frac{Ak}{\lambda\tau},
\]
this is precisely
\[
  \frac{d\alpha^*}{dn}
  = -\frac{k'}{k^2}\ln\!\frac{Ak}{\lambda\tau}
    + \frac{1}{k}\Bigl(\frac{k'}{k}-\frac{\lambda'}{\lambda}\Bigr).
\]

At the boundary $Ak=\lambda\tau$ we have $\ln(Ak/(\lambda\tau))=0$, so the first term vanishes and
\[
  \mathrm{sign}\!\left(\frac{d\alpha^*}{dn}\right)
  = \mathrm{sign}\!\left(\frac{k'}{k}-\frac{\lambda'}{\lambda}\right)
\]
at that point. More generally, away from the boundary the sign of $d\alpha^*/dn$ is governed by the inequality
\[
  \frac{k'}{k}-\frac{\lambda'}{\lambda}
  \;\ge\; \frac{k'}{k}\,\ln\!\frac{Ak}{\lambda\tau},
\]
which is the stated condition for $\alpha^*$ to increase in $n$.
\end{proof}

\begin{proof}[Proof of Lemma \ref{lem:min-spam}]
The attacker's expected profit is
\[
  \Pi_{\rm spam}(s)
  = \theta_b(s)\,R_x - s\,c_{\rm DA}.
\]
If $\theta_b$ is differentiable and strictly concave on $[0,\infty)$, then $\Pi_{\rm spam}$ is also strictly concave on $[0,\infty)$. Any interior maximizer $s^\ast>0$ must satisfy the first-order condition
\[
  \Pi_{\rm spam}'(s^\ast)
  = \theta_b'(s^\ast)\,R_x - c_{\rm DA} = 0,
\]
and strict concavity guarantees that such an $s^\ast$, when it exists, is unique. A profitable attack with some $s>0$ exists if and only if $\max_{s\ge 0} \Pi_{\rm spam}(s)>0$, which under $\theta_b(0)=0$ is equivalent to $\max_{s\ge 0}\theta_b'(s)\,R_x>c_{\rm DA}$.

If $\theta_b$ is piecewise linear with a single capacity cliff at some $s_{\rm cliff}$, then $\theta_b(s)$ is flat for $s<s_{\rm cliff}$ and jumps when $s$ crosses $s_{\rm cliff}$. For $s<s_{\rm cliff}$ we have $\theta_b(s)=0$ and $\Pi_{\rm spam}(s)=-s\,c_{\rm DA}$, so spam is unprofitable. For $s\ge s_{\rm cliff}$ the victim block $b$ is pushed below rank $L$ and $\theta_b(s)$ is constant at its post-cliff value, the minimal profitable spam is then precisely the smallest $s\ge s_{\rm cliff}$ such that $\theta_b(s)\,R_x - s\,c_{\rm DA}>0$.
\end{proof}

\begin{proof}[Proof of Proposition \ref{prop:ordering-spam}]
Consider a user who wants to move a target transaction with tip $\tau^\ast$ forward by $K$ positions in the merged execution order. Under tip priority with dependencies, each sacrificial transaction used for this purpose must pay the base fee $f$ and carry a tip strictly higher than the current frontier tip that it needs to beat. Let the minimal overbid above the current $k$-th order-statistic tip be $\Delta_k\ge 0$ for $k=1,\dots,K$. Then the $k$-th sacrificial transaction has tip $\tau^\ast+\Delta_k$ and costs $f+\tau^\ast+\Delta_k$ in total.

The total private cost of using $K$ such sacrificial transactions is therefore
\[
  C_K
  = \sum_{k=1}^K \bigl(f + \tau^\ast + \Delta_k\bigr),
\]
as stated. The private benefit from advancing the target forward by $K$ positions is bounded above by the incremental surplus $B_K$ obtained from moving earlier in the queue, and by definition $B_K\le W$, the first-mover surplus. Ordering spam is privately unprofitable whenever
\(
  C_K > B_K,
\)
since in that region the net gain $B_K-C_K$ is strictly negative. In particular, if $f>0$ then $C_K$ grows at least linearly in $K$, while $B_K$ is bounded above by $W$, so beyond a finite $K$ the inequality $C_K>B_K$ must hold regardless of the exact $\Delta_k$.
\end{proof}

\begin{proof}[Proof of Lemma \ref{lem:k-unique}]
For the private problem define the discrete increments
\[
  \Delta_k^{\rm priv}
  := U^{\rm priv}(k+1)-U^{\rm priv}(k)
  = v\bigl(\Psi_{k+1}-\Psi_k\bigr) - c.
\]
By assumption the marginal gains $\Psi_{k+1}-\Psi_k$ are strictly decreasing in $k$, hence the sequence $(\Delta_k^{\rm priv})_{k\in\mathbb N_0}$ is strictly decreasing in $k$. The sequence $(U^{\rm priv}(k))_{k\in\mathbb N_0}$ therefore increases as long as $\Delta_k^{\rm priv}>0$ and decreases once $\Delta_k^{\rm priv}\le 0$, so it is strictly unimodal and has a unique maximizer $k^{\rm priv}$.

For the social problem,
\[
  U^{\rm soc}(k)
  = v\,\Psi_k - k\,c - \eta\,e(k),
\]
the discrete increments are
\[
  \Delta_k^{\rm soc}
  := U^{\rm soc}(k+1)-U^{\rm soc}(k)
  = v\bigl(\Psi_{k+1}-\Psi_k\bigr) - c - \eta\bigl(e(k+1)-e(k)\bigr).
\]
Convexity of $e(\cdot)$ implies that $e(k+1)-e(k)$ is weakly increasing in $k$, so $(\Delta_k^{\rm soc})_{k\in\mathbb N_0}$ is again strictly decreasing in $k$. The same unimodality argument yields a unique maximizer $k^{\rm soc}$ for the social objective.
\end{proof}

\begin{proof}[Proof of Proposition \ref{prop:pigou}]
For the private problem, the discrete marginal gain from increasing $k$ to $k+1$ is
\[
  \Delta_k^{\rm priv}
  := U^{\rm priv}(k+1)-U^{\rm priv}(k)
  = v\bigl(\Psi_{k+1}-\Psi_k\bigr) - c.
\]
In an interior optimum $k^{\rm priv}$ (i.e.\ when the optimum is not at the boundary $k=0$), the necessary and sufficient first order condition for optimality in this discrete setting is
\[
  \Delta_{k^{\rm priv}-1}^{\rm priv} \;\ge\; 0,
  \qquad
  \Delta_{k^{\rm priv}}^{\rm priv} \;\le\; 0.
\]

For the social problem, the corresponding marginal gain is
\[
  \Delta_k^{\rm soc}
  := U^{\rm soc}(k+1)-U^{\rm soc}(k)
  = v\bigl(\Psi_{k+1}-\Psi_k\bigr) - c - \eta\bigl(e(k+1)-e(k)\bigr).
\]
At an interior social optimum $k^{\rm soc}$ we similarly have
\[
  \Delta_{k^{\rm soc}-1}^{\rm soc} \;\ge\; 0,
  \qquad
  \Delta_{k^{\rm soc}}^{\rm soc} \;\le\; 0.
\]

Comparing the two marginal conditions, we see that for each $k$
\[
  \Delta_k^{\rm soc}
  = \Delta_k^{\rm priv} - \eta\bigl(e(k+1)-e(k)\bigr)
  \;\le\; \Delta_k^{\rm priv},
\]
because $e(k+1)-e(k)\ge 0$ by monotonicity of $e$. Thus the social marginal gain crosses zero weakly earlier in $k$ than the private marginal gain. Equivalently, the $k$ that solves the social first order condition is weakly smaller than the one that solves the private first-order condition, so $k^{\rm priv}\ge k^{\rm soc}$.

To implement the social optimum via a Pigouvian surcharge, note that at an interior social optimum $k^{\rm soc}$ we want the private problem with surcharge $\psi_k$ per submission to have the same first-order condition:
\[
  v\bigl(\Psi_{k+1}-\Psi_k\bigr) = c + \psi_k
\]
evaluated at $k=k^{\rm soc}$. Comparing with the social first-order condition
\[
  v\bigl(\Psi_{k+1}-\Psi_k\bigr)
  = c + \eta\bigl(e(k+1)-e(k)\bigr)
\]
at $k^{\rm soc}$ shows that the required marginal surcharge at that point is
\[
  \psi_{k^{\rm soc}}
  = \eta\bigl(e(k^{\rm soc}+1)-e(k^{\rm soc})\bigr).
\]
In other words, a per-submission surcharge
\(
  \psi_k = \eta\bigl(e(k+1)-e(k)\bigr)
\)
exactly internalizes the marginal externality at $k$ and makes the private and social first order conditions coincide. In a smooth approximation with differentiable $e(\cdot)$, this discrete difference is captured by the derivative $\psi(k)=\eta e'(k)$.
\end{proof}

\begin{proof}[Proof of Theorem \ref{thm:pdm}]
The first two properties follow from standard results on lexicographically keyed topological sorts with a fixed total tie-breaker. The PDM algorithm is an instance of Kahn's algorithm \cite{clrs2009}, it repeatedly extracts an in-degree-zero vertex from a priority queue keyed by $\pi(\cdot)$ and removes it from the DAG. When $G_t$ is acyclic, Kahn's algorithm processes all vertices and outputs a linear extension of $G_t$, i.e.\ an order that respects all dependency edges. Since the priority key $\pi(\cdot)$ together with $h(\cdot)$ induces a total order on eligible vertices at each step, the output of PDM is unique, so determinism and feasibility follow.

For monotonicity in $\tau_x$, fix a transaction $x$ and consider running PDM twice on the same DAG $G_t$, once with the original tips $\{\tau_y\}_{y\in\mathcal X_t}$ and once with a strictly higher tip $\tau'_x>\tau_x$ for $x$ while keeping all other fields (including all other tips) unchanged. Before $x$ becomes eligible (in-degree zero), its tip does not affect the algorithm, so the sequence of steps until eligibility is identical in both runs. From the moment $x$ enters the heap on, its priority key $\pi(x)$ is weakly higher under $\tau'_x$ than under $\tau_x$, while the keys of all other eligible transactions are the same in the two runs. Thus at any step in which $x$ belongs to the heap in the original run, it cannot be extracted later in the modified run: increasing $\tau_x$ can only cause $x$ to be chosen earlier relative to other eligible vertices, never later. Eligibility itself is determined purely by $G_t$ and is therefore identical in the two runs. Consequently, the final position of $x$ in the output order moves (weakly) earlier when $\tau_x$ is increased.
\end{proof}

\section{Notation}\label{sec:not}
\begin{center}
\begin{tabular}{ll}
\hline
Symbol & Meaning \\
\hline
\(n\) & Number of proposers \\
\(t\) & Tick index \\
\(x\) & Transaction \\
\(\tau_x\) & Per\-tx tip (to proposer) \\
\(A\) & MEV amplitude/scale for a tx family \\
\(k\) & Opportunity accrual rate for MEV on \(x\) (per tick) \\
\(\lambda\) & Hazard that \(x\) gets pre\-empted/executes elsewhere (per tick) \\
\(B_i^t\) & Block (or set of txs) proposed by \(i\) in tick \(t\) \\
\(\Delta_x(\alpha)\) & Expected incremental MEV from delaying \(x\) by \(\alpha\) ticks \\
\(M(\tau)\) & Delay envelope \(M(\tau)=\sup_{\alpha\ge0}U_{\rm mev}(\alpha,\tau)\) \\
\(m\) & Number of bidders in an auction (excludes the seller) \\
\(\ell\) & PoA co\-signatures required: \(\ell=f{+}1\) (Gamma shape) \\
\(\mu_i\) & Rate of co\-signer arrivals for proposer \(i\) (PoA) \\
\(\phi_j^t\) & Shadow value of a marginal tx slot for proposer \(j\) in tick \(t\) \\
\hline
\end{tabular}
\end{center}
\section{Two-hazard generalization}
\label{app:twohazard}
For completeness, let the survival hazard be $\lambda>0$ and the tip-discount hazard be $\delta\ge 0$. Then
\begin{equation}
U_{\mathrm{mev}}(\alpha,\tau)
=\frac{A k}{k+\lambda}\bigl(1-e^{-(k+\lambda)\alpha}\bigr)
+\tau e^{-\delta\alpha}.
\end{equation}
If $\delta<k+\lambda$ and $0< r:=\tfrac{\delta\tau}{A k}<1$, the unique interior maximizer is
\begin{equation}
\alpha^{\ast}(\tau)
=\frac{1}{k+\lambda-\delta}
\ln\!\left(\frac{A k}{\delta\tau}\right),
\end{equation}
and
\begin{equation}
M(\tau)
=\frac{A k}{k+\lambda}\!\left(1-r^{\frac{k+\lambda}{k+\lambda-\delta}}\right)
+\tau\, r^{\frac{\delta}{k+\lambda-\delta}}.
\end{equation}
Otherwise,
\begin{equation}
M(\tau)=\max\!\left\{\tau,\frac{A k}{k+\lambda}\right\}.
\end{equation}
Setting $\delta=\lambda$ recovers the earlier model.
\end{document}